\documentclass{llncs}
\usepackage[utf8]{inputenc}
\usepackage{graphicx}
\usepackage{tabularx}
\usepackage{array}
\usepackage{pifont}
\usepackage{xspace}
\usepackage{todonotes}
\usepackage{paralist}
\usepackage{subfigure}
\usepackage{thm-restate}
\usepackage{lineno}
\newcommand{\NP}{\textit{NP}\xspace}

\newcommand{\ourproblemlong}[1]{{\sc $#1$-Clique2Path Planarity}}
\newcommand{\ourproblem}[1]{{\sc $#1$-C2PP}}

\let\doendproof\endproof
\renewcommand\endproof{~\hfill$\qed$\doendproof}

\begin{document}
	
\title{Turning Cliques into Paths to Achieve Planarity\thanks{This work began at the Bertinoro Workshop on Graph Drawing 2018. Research was partially supported by DFG grant Ka812/17-1 and MIUR-DAAD Joint Mobility Program n.57397196 (PA), by Young Scholar Fund/AFF - Univ. Konstanz (KK), by NSF grants CCF-1740858 - CCF-1712119 (SK), by project ``Algoritmi e sistemi di analisi visuale di reti complesse e di grandi dimensioni" - Ric. di Base 2018, Dip. Ingegneria - Univ. Perugia (GL, AT), by project GEO-SAFE n.H2020-691161 (AN).}
}
\author{Patrizio Angelini\inst{1}, Peter Eades\inst{2}, Seok-Hee Hong\inst{2}, \\ Karsten Klein\inst{3}, Stephen Kobourov\inst{4}, Giuseppe Liotta\inst{5}, \\ Alfredo Navarra\inst{5}, Alessandra Tappini\inst{5}}

\date{}

\institute{ University of T\"ubingen, Germany
\email{angelini@informatik.uni-tuebingen.de}
\and 
The University of Sydney, Australia
\email{\{peter.eades,seokhee.hong\}@usyd.edu.au}
\and University of Konstanz, Germany
\email{karsten.klein@uni-konstanz.de} 
\and
University of Arizona, USA
\email{kobourov@cs.arizona.edu}
\and University of Perugia, Italy
\email{\{giuseppe.liotta,alfredo.navarra\}@unipg.it}
\email{alessandra.tappini@studenti.unipg.it}
}

\maketitle


\begin{abstract}
Motivated by hybrid graph representations, we introduce and study the following beyond-planarity problem, which we call \ourproblemlong{h}: Given a graph $G$, whose vertices are partitioned into subsets of size at most $h$, each inducing a clique, remove edges from each clique so that the subgraph induced by each subset is a path, in such a way that the resulting subgraph of $G$ is planar.
We study this problem when $G$ is a simple topological graph, and establish its complexity in relation to $k$-planarity. 
We prove that \ourproblemlong{h} is NP-complete even when $h=4$ and $G$ is a simple $3$-plane graph, while it can be solved in linear time, for any $h$, when $G$ is $1$-plane.
\end{abstract}

\section{Hybrid Representations}

A common problem in the visual analysis of real-world networks is that dense subnetworks create occlusions and hairball-like structures in node-link diagrams generated by standard layout algorithms, e.g., force-directed methods. On the other hand, different representations, such as adjacency matrices, are well suited for dense graphs but make neighbor identification and path-tracing more difficult~\cite{DBLP:journals/ivs/GhoniemFC05,DBLP:conf/gd/OkoeJK17}. 
\emph{Hybrid graph representations} combine different representation meta\-phors in order to exploit their strengths and overcome their drawbacks. 

The first example of hybrid representation was the \emph{NodeTrix} model~\cite{HFM07}, which combines node-link diagrams with adjacency-matrix representations of the denser subgraphs~\cite{DBLP:journals/jgaa/LozzoBFP18,DBLP:conf/gd/GiacomoLPT17,HFM07,YANG17}. Another example of hybrid representations are  \emph{intersection-link representations}~\cite{DBLP:journals/jgaa/AngeliniLBFPR17}. In this model vertices are geometric objects and edges are either intersections between  objects (\emph{intersection edges}), or crossing-free Jordan arcs attaching at their boundary (\emph{link edges}).
Different types of  objects determine different intersection-link representations. 

In~\cite{DBLP:journals/jgaa/AngeliniLBFPR17}, \emph{clique-planar} drawings are defined as intersection-link representations in which the objects are isothetic rectangles, and the partition into intersection- and link-edges is given in the input, so that the graph induced by the intersection-edges is composed of a set of vertex-disjoint cliques. The corresponding recognition problem is called {\sc Clique-planarity}, and it has been proved NP-complete in general and polynomial-time solvable in restricted cases.

We study {\sc Clique-planarity} when all cliques have bounded size. As proved in~\cite{DBLP:journals/jgaa/AngeliniLBFPR17}, the {\sc Clique-planarity} problem can be reformulated in the terminology of \emph{beyond-planarity}~\cite{DBLP:journals/corr/abs-1804-07257,DBLP:journals/csr/KobourovLM17}, as follows. Given a graph $G=(V,E)$ and a partition of its vertex set $V$ into subsets $V_1,\dots,V_m$ such that the subgraph of $G$ induced by each subset $V_i$ is a clique, the goal is to compute a planar subgraph $G'=(V,E')$ of $G$ by replacing the clique induced by $V_i$, for each $i = 1,\dots,m$, with a path spanning the vertices of $V_i$. We call \ourproblemlong{h} (for short, \ourproblem{h}) the version of this problem in which each clique has size at most~$h$;~see~Fig.~\ref{fig:example}.

We remark that the version of \ourproblem{h} in which the input graph $G$ is a \emph{geometric graph}, i.e., it is drawn in the plane with straight-line edges, has been recently studied by Kindermann et al.~\cite{kkrss-tpsfp-eurocg18} in a different context. The input of their problem is a set of colored points in the plane, and the goal is to decide whether there exist straight-line spanning trees, one for each same-colored point subset, that do not cross each other. Since edges are straight-line, their drawings are determined by the positions of the points, and hence each same-colored point subset can in fact be seen as a straight-line drawing of a clique, from which edges have to be removed so that each clique becomes a tree and the drawing becomes planar. They proved NP-completeness for the case in which the spanning tree must be a path, even when there are at most $4$ vertices with the same color. This implies that \ourproblem{4} for geometric graphs in NP-complete. On the other hand, they provided a linear-time algorithm when there exist at most $3$ vertices with the same color, which then extends to \ourproblem{3} for geometric graphs.

In this paper, we study the version of \ourproblem{h} in which the input graph $G$ is a \emph{simple topological graph}, that is, it is embedded in the plane so that each edge is a Jordan arc connecting its end-vertices; by simple we mean that a Jordan arc does not pass through any vertex, and does not intersect any arc more than once (either with a proper crossing or sharing a common end-vertex); finally, no three arcs pass through the same point. 
Our main goal is to study the complexity of this problem in relation to the well-studied class of \emph{$k$-planar graphs}, i.e., those that admit drawings in which each edge has at most $k$ crossings~\cite{DBLP:journals/jgaa/AngeliniLBFPR17,DBLP:conf/compgeom/Bekos0R17,DBLP:journals/corr/abs-1804-07257,DBLP:journals/combinatorica/PachT97}.

We observe that the NP-completeness of \ourproblem{4} for geometric graphs already implies the NP-completeness of \ourproblem{4} for simple topological graphs; also, though not explicitly mentioned in~\cite{kkrss-tpsfp-eurocg18}, it is possible to show that the instances produced by that reduction are $4$-plane (see Appendix~\ref{app:eurocg}). 
We strengthen this result by proving in Section~\ref{se:hardness} that \ourproblem{4} is NP-complete even for simple topological $3$-plane graphs. On the positive side, we prove in Section~\ref{se:polynomial} that the \ourproblem{h} problem for simple topological $1$-plane graphs can be solved in linear time for any value of $h$. We finally remark that the 2-SAT formulation used in~\cite{kkrss-tpsfp-eurocg18} to solve \ourproblem{3} for geometric graphs can be easily extended to solve \ourproblem{3} for any simple topological graph.

For space reasons, some proofs have been omitted or sketched, and can be found in the appendix; the corresponding statements are marked with [*].

\section{\NP-completeness for simple topological $3$-plane graphs} \label{se:hardness}

In this section we prove that the \ourproblem{k} problem remains NP-complete for $k=4$ even when the input is a simple topological $3$-plane graph.  

Since the planarity of a simple topological graph can be checked in linear time, the \ourproblem{h} problem for simple topological $k$-plane graphs belongs to NP for all values of $h$ and~$k$. 
In the following, we prove the NP-hardness by means of a reduction from the {\sc Planar Positive 1-in-3-SAT} problem. In this version of the {\sc Satisfiability} problem, which is known to be NP-complete~\cite{DBLP:journals/jacm/MulzerR08}, each variable appears only with its positive literal, each clause has at most three variables, the graph obtained by connecting each variable with all the clauses it belongs to is planar, and the goal is to find a truth assignment in such a way that, for each clause, exactly one of its three variables is set to \texttt{True}.
For each $3$-clique we use in the reduction, there is a \emph{base edge}, which is crossing-free in the constructed topological graph, while the other two edges always have crossings. We call \emph{left} (\emph{right}) the edge that follows (precedes) the base edge in the clockwise order of the edges along the $3$-clique. Also, if an edge $e$ of a clique does not belong to the path replacing the clique, we say that~$e$ is \emph{removed}, and that all the crossings involving $e$ in $G$ are \emph{resolved}.
For each variable $x$, let $n_x$ be the number of clauses containing $x$. We construct a simple topological graph gadget $G_x$ for $x$, called \emph{variable gadget}; see the left dotted box in Fig.~\ref{fig:variable-gadget}. This gadget contains $2n_x$ $3$-cliques $t^x_1,\dots,t^x_{2n_x}$, forming a ring, so that the left (right) edge of $t^x_i$ only crosses the left (right) edge of $t^x_{i-1}$ and of $t^x_{i+1}$, for each $i=1,\dots,2n_x$. Also, gadget $G_x$ contains $n_x$ additional $3$-cliques, called $\tau^x_1,\dots,\tau^x_{n_x}$, so that the right edge of $\tau^x_j$ crosses the left edge of $t^x_{2j-1}$ and the right edge of $t^x_{2j}$, while the left edge of $\tau^x_j$ crosses the left edge of $t^x_{2j}$ and the right edge of $t^x_{2j-1}$.
\begin{figure}[tb]
	\centering
    \subfigure[\label{fig:variable-gadget}{}]
	{\includegraphics[width=.47\textwidth,page=1]{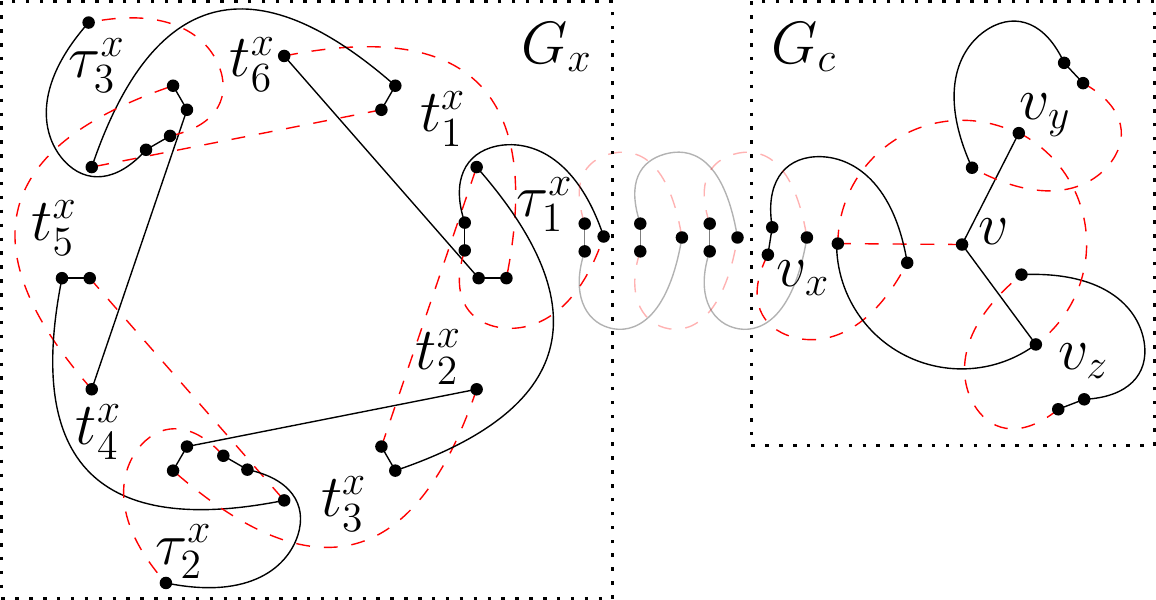}}
	\hfill
	\subfigure[\label{fig:clauses-all-false}{}]
	{\includegraphics[width=.16\textwidth,page=2]{variable-gadget.pdf}}
	\hfill
	\subfigure[\label{fig:clauses-two-true}{}]
	{\includegraphics[width=.16\textwidth,page=3]{variable-gadget.pdf}}
\caption{(a) The variable gadget $G_x$ for a variable $x$ is represented in the left dotted box. The clause gadget for a clause $c$ is represented in the right dotted box. The chain connecting $G_x$ to $G_c$ is represented with lighter colors. The removed edges are dashed red. (b) All variables are \texttt{False}. (c) At least two variables are \texttt{True}.}
	\label{fig:clauses}
\end{figure}
Then, for each clause $c$, we construct a topological graph gadget $G_c$, called \emph{clause gadget}, which is composed of a planar drawing of a $4$-clique, together with three $3$-cliques whose left and right edges cross the edges of the $4$-clique as in the right dotted box in Fig.~\ref{fig:variable-gadget}. In particular, observe that the right (left) edge of each $3$-clique crosses exactly one (two) edges of the $4$-clique.
Every $3$-clique in $G_c$ corresponds to one of the three variables of $c$. Let $x$ be one of such variables; assuming that $c$ is the $j$-th clause that contains $x$ according to the order of the clauses in the given formula, we connect the $3$-clique corresponding to $x$ in the clause gadget $G_c$ to the $3$-clique $\tau^x_j$ of the variable gadget $G_x$ of $x$ by a chain of $3$-cliques of odd length, as in Fig.~\ref{fig:variable-gadget}.

By construction, the resulting simple topological graph $G$ contains cliques of size at most $4$, namely one per clause, and hence is a valid instance of \ourproblem{4}. Also, by collapsing each variable and clause gadget into a vertex, and each chain connecting them into an edge, the resulting graph $G'$ preserves the planarity of the {\sc Planar Positive 1-in-3-SAT} instance. This implies that the only crossings for each edge of $G$ are with other edges in the gadget it belongs to and, possibly, with the edges of the $3$-cliques of a chain. Hence, $G$ is $3$-plane. Namely, each base edge is crossing-free; each internal edge of a $4$-clique has one crossing; each external edge of a $4$-clique has two crossings, and the same is true for the left and right edges of each $3$-clique in a chain; finally, the left and right edges of each $3$-clique in either a variable or a clause gadget has three crossings.

In the following we prove the equivalence between the original instance of {\sc Planar Positive 1-in-3-SAT} and the constructed instance $G$ of \ourproblem{4}. For this, we first give a lemma stating that variable gadgets correctly represent the behavior of a variable;  indeed they can assume one out of two possible states in any solution for \ourproblem{4}. The proof of the next lemma is in Appendix~\ref{app:proof-lemmavariable}.

\begin{restatable}{lemma}{lemmavariable}\label{le:variable} \emph{\textbf{[*]}}
Let $G_x$ be the variable gadget for a variable $x$ in $G$. Then, in any solution for \ourproblem{4}, either the left edge of each $3$-clique $\tau^x_j$, with $j=1,\dots,n_x$, is removed, or the right edge of each $3$-clique $\tau^x_j$ is removed.
\end{restatable}

Given Lemma~\ref{le:variable}, we can associate the truth value of a variable $x$ with the fact that either the left or the right edge of each $3$-clique $\tau^x_j$ in the variable gadget $G_x$ of $G$ is removed. We use this association to prove the following theorem.

\begin{restatable}{theorem}{theoremourproblemfournpcomplete}\label{th:ourproblem-4-np-complete} \emph{\textbf{[*]}}
	The \ourproblem{4} problem is \NP-complete, even for $3$-plane graphs.
\end{restatable}

\begin{proof}[sketch]
Given an instance of {\sc Planar Positive 1-in-3-SAT}, we construct an instance $G$ of \ourproblem{4} in linear time as described above. We prove one direction of the equivalence between the two problems. The other direction follows a similar reasoning.
Suppose that there exists a solution for \ourproblem{4}, i.e., a set of edges of $G$ whose removal resolves all  crossings. By Lemma~\ref{le:variable}, for each variable $x$ either the left or the right edge of each $3$-clique $\tau^x_j$ in gadget $G_x$ is removed. We assign \texttt{True} (\texttt{False}) to $x$ if the right (left) edge is removed.  

We first claim that for each clause $c$ that contains variable $x$, the right (left) edge of the $3$-clique $t_c(x)$ of the clause gadget $G_c$ corresponding to $x$ is removed if and only if the right (left) edge of each $3$-clique $\tau^x_j$ is removed. Consider the chain that connects $t_c(x)$ with a $3$-clique $\tau^x_j$ of $G_x$. For any two consecutive $3$-cliques along the chain the left edge of one $3$-clique and the right edge of the other $3$-clique must be removed. Since the chain has odd length, the truth value of $G_x$ is transferred to the $3$-clique $t_c(x)$ of $G_c$ and thus the claim follows. 

Consider now a clause $c$ with variables $x$, $y$, and $z$. Let $t_c(x)$, $t_c(y)$, and $t_c(z)$ be the $3$-cliques of the clause gadget $G_c$ of $c$ corresponding to $x$, $y$, and $z$, respectively. Let $v$ be the central vertex of the $4$-clique of $G_c$, and let $v_x$, $v_y$, $v_z$ be the vertices of this $4$-clique lying inside $t_c(x)$, $t_c(y)$, and $t_c(z)$ (see Fig.~\ref{fig:clauses}). Assume that $v_x$, $v_y$, and $v_z$ appear in this clockwise order around $v$. 
We now show that, for exactly one of $t_c(x)$, $t_c(y)$, and $t_c(z)$ the right edge is removed, which implies that exactly one of $x$, $y$, and $z$ is \texttt{True} and hence the instance of {\sc Planar Positive 1-in-3-SAT} is positive.
Assume that for each of $t_c(x)$, $t_c(y)$, and $t_c(z)$ the left edge is removed (i.e., all the three variables are set to \texttt{False}), as in Fig.~\ref{fig:clauses-all-false}. The crossings between the right edges of the three $3$-cliques and the three edges of triangle $(v_x,v_y,v_z)$ are not resolved. All edges of this triangle should be removed, which is not possible since the remaining edges of the $4$-clique do not form a path.
Assume now that for at least two of the $3$-cliques, say $t_c(x)$ and $t_c(y)$, the right edge is removed (i.e., $x$ and $y$ are set to \texttt{True}), as in Fig.~\ref{fig:clauses-two-true}. Since each edge of triangle $(v_x,v_y,v)$ is crossed by the left edge of one of $t_c(x)$ and $t_c(y)$, by construction, these crossings are not resolved. Hence, all edges of $(v_x,v_y,v)$ should be removed, which is not possible since the remaining edges of the $4$-clique do not form a path of length $4$.
Finally, assume that for exactly one of the $3$-cliques, say $t_c(x)$, the right edge is removed (i.e., $x$ is the only one set to \texttt{True}), as in Fig.~\ref{fig:variable-gadget}. By removing edges $(v,v_x)$, $(v_x,v_y)$, and $(v_y,v_z)$, all crossings are resolved; the remaining edges of the $4$-clique form a path of length $4$, as desired.
\end{proof}

\section{\ourproblemlong{h} and $1$-Planarity} \label{se:polynomial}

In this section we show that, when the given simple topological graph is 1-plane, problem \ourproblem{h} can be solved in linear time in the size of the input, for any $h$. We consider all possible simple topological 1-plane cliques and show that the problem can be solved using only local tests, each requiring constant time. Note that $h\leq 6$, since $K_6$ is the largest $1$-planar complete graph~\cite{DBLP:journals/csr/KobourovLM17}.

Simple topological $1$-plane graphs containing cliques with at most four vertices that cross each other can be constructed, but it is easy to enumerate all these graphs (up to symmetry); 
see Fig.~\ref{fig:34comb}. Note that such graphs involve at most two cliques and that if $K_4$ has a crossing, combining it with any other clique would violate $1$-planarity; see Fig.~\ref{fi:cliques-a} and Fig.~\ref{fi:cliques-b}. The next lemma accounts for cliques with five or six vertices.

\begin{figure}[tb]
	\centering
    	\subfigure[\label{fi:cliques-a}] {\includegraphics[width=0.13\textwidth,page=2]{3_4-cliques}}
	\hfil	\subfigure[\label{fi:cliques-b}] {\includegraphics[width=0.13\textwidth,page=3]{3_4-cliques}}
	\hfill
	\subfigure[\label{fi:cliques-inv-a}] {\includegraphics[width=0.105\textwidth,page=1]{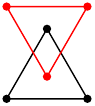}}
	\hfill
\subfigure[\label{fi:cliques-inv-b}] {\includegraphics[width=0.105\textwidth,page=2]{3_4-cliques-inv}}
	\hfill
\subfigure[\label{fi:cliques-inv-c}] {\includegraphics[width=0.105\textwidth,page=3]{3_4-cliques-inv}}
	\hfill
\subfigure[\label{fi:cliques-inv-d}] {\includegraphics[width=0.105\textwidth,page=4]{3_4-cliques-inv}}
	\hfill
\subfigure[\label{fi:cliques-inv-e}] {\includegraphics[width=0.115\textwidth,page=5]{3_4-cliques-inv}}
\caption{All 1-plane graphs involving one or more cliques of type $K_3$ and $K_4$.} \label{fig:34comb}
\end{figure}

\begin{lemma}\label{lem::cliquefive}
There exists no $1$-plane simple topological graph that contains two cliques, one of which with at least five vertices, whose edges cross each other.
\end{lemma}

\begin{proof}
Consider a simple $1$-plane graph $G$ that contains two disjoint cliques $K$ and $H$, with five and three vertices, respectively. Let $K'$ be the simple plane topological graph obtained from $K$ by replacing each crossing with a dummy vertex. By $1$-planarity, every face of $K'$ is a triangle and contains at most one dummy vertex. 
Suppose, for a contradiction, that there exists a crossing between an edge of $K$ and an edge of $H$ in $G$. Then there would exist at least a vertex $v$ of $H$ inside a face $f$ of $K'$ and at least one outside $f$. Since $H$ is a triangle, there must have been two edges that connect vertices inside $f$ to vertices outside $f$.
If $f$ contains one dummy vertex, then two of its edges are not crossed by edges of $H$, as otherwise $G$ would not be $1$-plane. Hence, both the edges that connect vertices inside $f$ to vertices outside $f$ cross the other edge of $f$, a contradiction. 
If $f$ contains no dummy vertices, then each edge of $f$ admits one crossing. Let $u$ be the vertex of $f$ that is incident to the two edges crossed by edges of $H$. Since $u$ has degree $4$ in $K$, it is not possible to draw the third edge of $H$ so that it crosses only one edge of $K$, which completes the proof.
\end{proof}

Combining the previous discussion with Lemma~\ref{lem::cliquefive}, we conclude that, for each subgraph of the input graph $G$ that consists either of a combination of at most two cliques of size at most $4$, as in Fig.~\ref{fig:34comb}, or of a single clique not crossing any other clique, the crossings involving this subgraph (possibly with other edges not belonging to cliques) can only be resolved by removing its edges, which can be checked in constant time. In the next theorem, $n$ denotes the number of vertices.

\begin{theorem}\label{th:P2}
	\ourproblem{h} is $O(n)$-time solvable for simple topological $1$-plane graphs. 
\end{theorem}

\section{Open Problems}\label{se:conclusions}

We studied the \ourproblemlong{h} problem for simple topological $k$-plane graphs; we proved that this problem is NP-complete for $h = 4$ and $k=3$, while it is solvable in linear time for every value of $h$, when $k=1$. The natural open question is: what is the complexity for simple topological $2$-plane graphs? 

Kindermann et al.~\cite{kkrss-tpsfp-eurocg18} recently proved that problem \ourproblem{4} is NP-complete for geometric $4$-plane graphs. It would be interesting to study this geometric version of the problem for $2$-plane and $3$-plane graphs.

Finally, note that the version of the \ourproblem{h} problem when the input is an abstract graph (which is equivalent to {\sc Clique Planarity}~\cite{DBLP:journals/jgaa/AngeliniLBFPR17}) is NP-complete when $h \in O(n)$. What if $h$ is bounded by a constant or a sublinear function? 
We remark that, for $h=3$, this version of the problem is equivalent to {\sc Clustered Planarity}, when restricted to instances in which the graph induced by each cluster consists of three isolated vertices.

\clearpage

\bibliographystyle{splncs04}
\bibliography{lit}

\clearpage
\appendix
\makeatletter
\noindent
\rlap{\color[rgb]{0.51,0.50,0.52}\vrule\@width\textwidth\@height1\p@}%
\hspace*{7mm}\fboxsep1.5mm\colorbox[rgb]{1,1,1}{\raisebox{-0.4ex}{%
\large\selectfont\bfseries Appendix}}%
\makeatother

\begin{figure}[h]
	\centering
	\subfigure[\label{fig:non-planar-graph}{}]
	{\includegraphics[width=.32\textwidth,page=3]{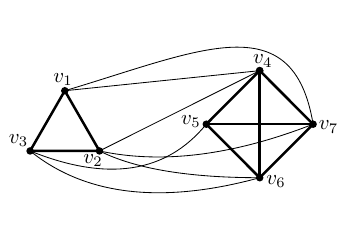}}
	\hfill 	
	\subfigure[\label{fig:inters-link}{}]
	{\includegraphics[width=.32\textwidth,page=4]{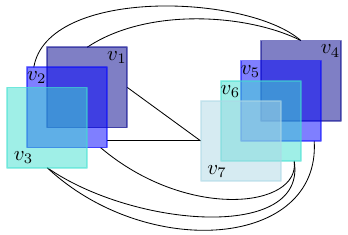}}
    \hfill 	
	\subfigure[\label{fig:h-c2pp}{}]
	{\includegraphics[width=.32\textwidth,page=5]{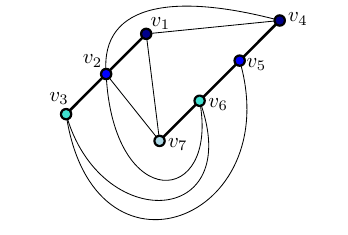}}
	\caption{(a)~A non-planar graph $G$. Cliques are highlighted with bold edges. (b)~A~clique-planar drawing of $G$. (c)~Replacing each clique with a path spanning its vertices. Note that differently from (a), in (c) the first vertex and the last vertex of each path have only one place to connect to edges, while the interior vertices have two places: This is what makes the problem non-trivial.}
	\label{fig:example}
\end{figure}

\section{Omitted Details About the Reduction in \cite{kkrss-tpsfp-eurocg18}}\label{app:eurocg}

In this section, we show that the instances produced by the reduction in~\cite{kkrss-tpsfp-eurocg18} are $4$-plane in general.

The variable gadget consists of a triangle $X$ whose edges are $x$, $x_l$ and $x_r$. Edge $x$ is crossing-free and the truth value of $X$ is encoded according to which edge among $x_l$ and $x_r$ is crossing-free.
Given a pair of triangles $T_1$ and $T_2$ whose vertices are ${u,y,z}$ and ${v,y,z}$, they define two faces $f_1$ is and $f_2$ respectively. Concatenate a triangle $T_3$ defined as in the variable gadget with $f_1$ by inserting its crossing-free edge inside $f_1$ and by crossing the other two edges of $T_3$ with $(u,y)$ and $(u,z)$, respectively. Now, concatenate another triangle $T_4$ defined as in the variable gadget with $f_2$. If the crossing-free edge of $T_4$ is inside $f_2$, the gadget composed of $T_1, T_2, T_3$ and $T_4$ is the wire gadget; if the crossing-free edge of $T_4$ is outside $f_2$, the gadget composed of $T_1, T_2, T_3$ and $T_4$ is the inverter gadget.
The splitting gadget consists of three variable gadgets $X,Y$ and $Z$, and two $4$-cliques, concatenated as illustrated inside the blue region in Fig.~\ref{fi:4cr}, where the yellow region contains a variable gadget, the orange region contains a wire gadget and the violet region contains an inverter gadget.
As shown in Fig.~\ref{fi:4cr}, multiple splittings of a variable $X$ lead to an instance where a triangle has two edges with four crossings.

\begin{figure}[tb]
	\centering
	\includegraphics[width=.8\textwidth]{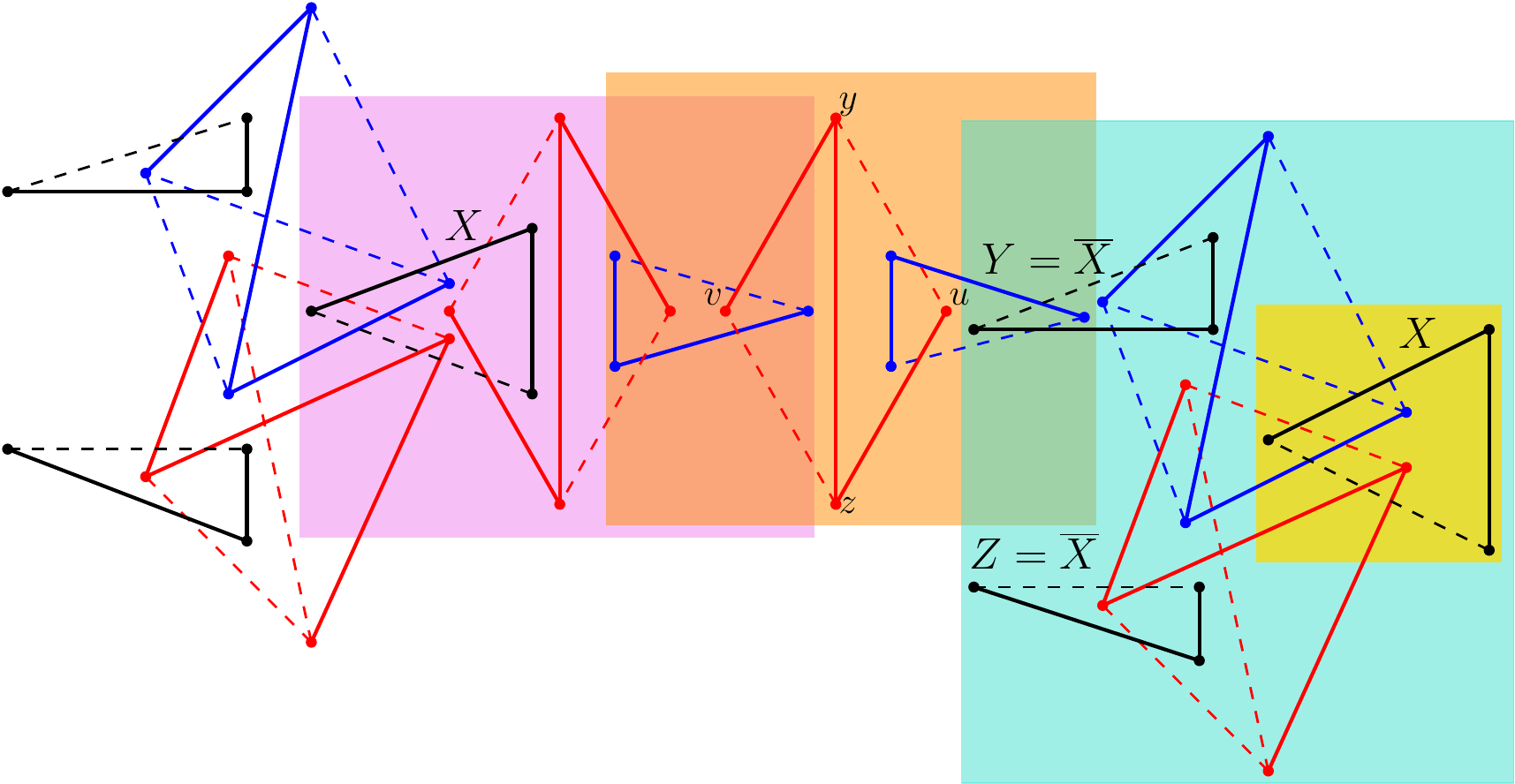}
	\caption{An instance given by the reduction in~\cite{kkrss-tpsfp-eurocg18}. The yellow region contains a variable gadget, the blue region contains a splitting gadget, the orange region contains a wire gadget and the violet region contains an inverter gadget.}\label{fi:4cr}
\end{figure}

\section{Proof of Lemma~\ref{le:variable}}\label{app:proof-lemmavariable}
\lemmavariable*
\begin{proof}
We first consider the possible removals of edges in $t^x_1,\dots,t^x_{2n_x}$ and claim that, in any solution for \ourproblem{4}, one of the two following conditions are satisfied:
\begin{inparaenum}[(i)]
\item\label{c:odd} for each $3$-clique $t^x_i$, if $i$ is odd, then the left edge is removed, while if $i$ is even the right edge is removed;
\item\label{c:even} for each $3$-clique $t^x_i$, if $i$ is odd, then the right edge is removed, while if $i$ is even the left edge is removed.
\end{inparaenum}
Note that this claim is sufficient to prove the statement; in fact, if condition~(\ref{c:odd}) holds (as in Fig.~\ref{fig:variable-gadget}), then the right edge of each $3$-clique $\tau^x_j$ must be removed, in order to resolve its crossings with the right edge of $t^x_{2j-1}$ and with the left edge of $t^x_{2j}$, while if condition~(\ref{c:even}) holds, then the left edge of each $3$-clique $\tau^x_j$ must be removed, in order to resolve its crossings with the left edge of $t^x_{2j-1}$ and with the right edge of $t^x_{2j}$.
	
In order to prove the claim, we consider the possible removals of edges of $t^x_1$. Suppose first that the base edge of $t^x_1$ is removed. Thus, the crossings between the left (right) edge of $t^x_1$ and the left (right) edge of $t^x_2$ are not resolved; this implies that they have to be resolved by removing both the left and the right edge of $t^x_2$, which is not possible. If the right edge of $t^x_1$ is removed, then the crossing between the right edges of $t^x_1$ and $t^x_2$ is resolved, while the one between their left edges is not. Hence, the left edge of $t^x_2$ must be removed. By iterating this argument we conclude that the right (left) edge of each $t^x_i$ with $i$ odd (even) is removed. Symmetrically, we can prove that, if the left edge of $t^x_1$ is removed, then the left (right) edge of each $t^x_i$ with $i$ odd (even) is removed. This concludes the proof of the lemma.
\end{proof}

\section{Proof of Theorem~\ref{th:ourproblem-4-np-complete}}\label{app:proof-theoremourproblem4npcomplete}
\theoremourproblemfournpcomplete*
\begin{proof}
Given an instance of {\sc Planar Positive 1-in-3-SAT}, we construct an instance $G$ of \ourproblem{4} in linear time as described above. We prove their equivalence.

Suppose first that there exists a solution for \ourproblem{4}, i.e., a set of edges of $G$ whose removal resolves all  crossings. By Lemma~\ref{le:variable}, for each variable $x$ either the left or the right edge of each $3$-clique $\tau^x_j$ in the variable gadget $G_x$ is removed. If the right edge is removed, we assign value \texttt{True} to variable $x$, otherwise we assign \texttt{False}.

In order to prove that this assignment results in a solution for the given formula of {\sc Planar Positive 1-in-3-SAT}, we first show that, for each clause $c$ that contains variable $x$, the right (left) edge of the $3$-clique $t_c(x)$ of the clause gadget $G_c$ corresponding to $x$ is removed if and only if the right (left) edge of each $3$-clique $\tau^x_j$ is removed. Namely, consider the chain that connects $t_c(x)$ with a $3$-clique $\tau^x_j$ of $G_x$. Note that, for any two consecutive $3$-cliques along the chain, the left edge of one $3$-clique and the right edge of the other $3$-clique must be removed. Since the chain has odd length, the right (left) edge of $t_c(x)$ is removed if and only if the right (left) edge of $\tau^x_j$ is removed, that is, the truth value of $G_x$ is transferred to the $3$-clique $t_c(x)$ of $G_c$. 

Finally, consider any clause $c$, composed of variables $x$, $y$, and $z$. Let $t_c(x)$, $t_c(y)$, and $t_c(z)$ be the three $3$-cliques of the clause gadget $G_c$ of $c$ corresponding to $x$, $y$, and $z$, respectively; also, let $v$ be the central vertex of the $4$-clique of $G_c$, and let $v_x$, $v_y$, $v_z$ be the vertices of this $4$-clique lying inside $t_c(x)$, $t_c(y)$, and $t_c(z)$, respectively; see Fig.~\ref{fig:clauses}. We assume w.l.o.g. that $v_x$, $v_y$, and $v_z$ appear in this clockwise order around $v$.
As discussed above, the left or the right edge of $t_c(x)$ (of $t_c(y)$; of $t_c(z)$) is removed depending on whether the left or the right edge of each $\tau^x_j$ (of each $\tau^y_j$; of each $\tau^z_j$) is removed. 
We show that, for exactly one of $t_c(x)$, $t_c(y)$, and $t_c(z)$ the right edge is removed, which then implies that exactly one of $x$, $y$, and $z$ is \texttt{True}, and hence the instance of {\sc Planar Positive 1-in-3-SAT} is positive.

Suppose first that for each of $t_c(x)$, $t_c(y)$, and $t_c(z)$ the left edge is removed (and hence all the three variables are set to \texttt{False}), as in Fig.~\ref{fig:clauses-all-false}. This implies that the crossings between the right edges of the three $3$-cliques and the three edges of triangle $(v_x,v_y,v_z)$ are not resolved. Hence, all the edges of this triangle should be removed, which is not possible since the remaining edges of the $4$-clique do not form a path.
Suppose now that for at least two of $t_c(x)$, $t_c(y)$, and $t_c(z)$, say $t_c(x)$ and $t_c(y)$, the right edge is removed (and hence $x$ and $y$ are set to \texttt{True}), as in Fig.~\ref{fig:clauses-two-true}. Since each edge of triangle $(v_x,v_y,v)$ is crossed by the left edge of one of $t_c(x)$ and $t_c(y)$, by construction, these crossings are not resolved. Hence, all the edges of $(v_x,v_y,v)$ should be removed, which is not possible since the remaining edges of the $4$-clique do not form a path of length $4$.
Suppose finally that for exactly one of $t_c(x)$, $t_c(y)$, and $t_c(z)$, say $t_c(x)$, the right edge is removed (and hence $x$ is the only one to be set to \texttt{True}), as in Fig.~\ref{fig:variable-gadget}. Then, by removing edges $(v,v_x)$, $(v_x,v_y)$, and $(v_y,v_z)$, all the crossings are resolved and the remaining edges of the $4$-clique form a path of length $4$, as desired.


	
The proof of the other direction is analogous. Namely, suppose that there exists a truth assignment that assigns a \texttt{True} value to exactly one variable in each clause. Then, for each variable $x$ that is set to \texttt{True} (to \texttt{False}), we remove the right (left) edge of each $3$-clique $t^x_i$, with $i=2j-1$ and $j = 1,\dots,n_x$, we remove the left (right) edge of each $3$-clique $t^x_i$, with $i=2j$ and $j = 1,\dots,n_x$, and we remove the right (left) edge of each $3$-clique $\tau^x_j$, with $j=1,\dots,n_x$. Then, we remove the left or right edge of each $3$-clique in a chain so that for any two consecutive $3$-cliques, one of them has been removed the left edge and the other one the right edge. This ensures that, for each clause $c$, the right edge of exactly one of the three $3$-cliques that belong to the clause gadget $G_c$ has been removed, say the one corresponding to variable $x$, while for the other two $3$-cliques the left edge has been removed. Hence, we can resolve all crossings by removing edges $(v,v_x)$, $(v_x,v_y)$, and $(v_y,v_z)$, as discussed above; see Fig.~\ref{fig:variable-gadget}. The statement follows.
\end{proof}

\end{document}